\documentclass[a4paper,UKenglish,cleveref, autoref, thm-restate,authorcolumns]{lipics-v2019}

\title{Process Symmetry in Probabilistic Transducers} 

\author{Shaull Almagor}{Computer Science Department, Technion, Israel }{shaull@cs.technion.ac.il}{https://orcid.org/0000-0001-9021-1175}{Supported by a European Union's Horizon 2020 research and innovation programme under the Marie Sk{\l}odowska-Curie grant agreement No 837327.}

\authorrunning{S. Almagor} 

\Copyright{S. Almagor} 

\ccsdesc[500]{Theory of computation~Verification by model checking}
\ccsdesc[300]{Theory of computation~Abstraction}
\ccsdesc[300]{Theory of computation~Concurrency}

\keywords{Symmetry, Probabilistic Transducers, Model Checking, Permutations} 

\category{} 

\relatedversion{} 

\supplement{}

\nolinenumbers 

\EventEditors{John Q. Open and Joan R. Access}
\EventNoEds{2}
\EventLongTitle{42nd Conference on Very Important Topics (CVIT 2016)}
\EventShortTitle{CVIT 2016}
\EventAcronym{CVIT}
\EventYear{2016}
\EventDate{December 24--27, 2016}
\EventLocation{Little Whinging, United Kingdom}
\EventLogo{}
\SeriesVolume{42}
\ArticleNo{23}

\usepackage{todonotes}
\usepackage{bbm}
\usepackage[basic]{complexity}

\newcommand{\IOo}{(2^{I\cup O})^\omega}

\newcommand{\Is}{(2^{I})^*}

\newcommand{\IOp}{(2^{I\cup O})^+}
\newcommand{\Ip}{(2^{I})^+}
\newcommand{\Op}{(2^{O})^+}

\newcommand{\tIO}{2^{I\cup O}}
\newcommand{\tI}{2^{I}}
\newcommand{\tO}{2^{O}}

\newcommand{\tup}[1]{\langle #1\rangle}
\newcommand{\ST}{\, :\, }
\renewcommand{\vec}[1]{{\mathbf{#1}}}

\newcommand{\cA}{\mathcal{A}}
\newcommand{\cB}{\mathcal{B}}
\newcommand{\cN}{\mathcal{N}}

\newcommand{\cT}{\mathcal{T}}
\newcommand{\cS}{\mathcal{S}}

\newcommand{\bbN}{\mathbb{N}}

\newcommand{\bbE}{\mathbb{E}}

\newcommand{\dirac}[1]{\mathbbm{1}[{#1}]}

\newcommand{\lab}{\boldsymbol{\ell}}
\newcommand{\dist}{\Delta}

\newcommand{\runs}{\mathsf{runs}}
\newcommand{\supp}{\mathsf{Supp}}

\newcommand{\sinit}{s_{\rm init}}
\newcommand{\smid}{s_{\rm mid}}

\newcommand{\num}{\boldsymbol{\#}}
\newcommand{\park}{\mathfrak{P}}
\newcommand{\rew}{\mathsf{R}}

\begin{document}

\maketitle

\begin{abstract}
Model checking is the process of deciding whether a system satisfies a given specification. Often, when the setting comprises multiple processes, the specifications are over sets of input and output signals that correspond to individual processes. Then, many of the properties one wishes to specify are symmetric with respect to the processes identities. 
In this work, we consider the problem of deciding whether the given system exhibits symmetry with respect to the processes' identities. When the system is symmetric, this gives insight into the behaviour of the system, as well as allows the designer to use only representative specifications, instead of iterating over all possible process identities.

Specifically, we consider probabilistic systems, and we propose several variants of symmetry. We start with precise symmetry, in which, given a permutation $\pi$, the system maintains the exact distribution of permuted outputs, given a permuted inputs. We proceed to study approximate versions of symmetry, including symmetry induced by small $L_\infty$ norm, variants of Parikh-image based symmetry, and qualitative symmetry. For each type of symmetry, we consider the problem of deciding whether a given system exhibits this type of symmetry.
\end{abstract}

\section{Introduction}
\label{sec:intro}
A fundamental approach to automatic verification is \emph{model checking}~\cite{clarke2018model}, where we are given a system and a specification, and we check whether all possible behaviours of the system satisfy the specification. In model checking of \emph{reactive} systems, the specification is over sets of inputs $I$ and outputs $O$, and the system is an $I/O$ transducer, which takes sequences of inputs in $\tI$, and responds with an output in $\tO$. Then, model checking amounts to deciding whether for every input sequence, the matching output sequence generated by the transducer, satisfies the specification.

In practice, and especially in verification of concurrent systems, the input and output sets have some correspondence. For example, in an arbiter for $k$ processes, the inputs are typically $I=\{i_1,\ldots,i_k\}$, where $i_j$ is interpreted as ``a request was generated by Process $j$'', and the outputs are $O=\{o_1,\ldots,o_k\}$, where $o_j$ is interpreted as ``Process $j$ was granted access''.
In such cases, specification often end up having symmetric repetitions of a similar pattern. For example, we may wish to specify that in our arbiter, if Process $j_1$ generated a request before Process $j_2$, then a grant for $j_1$ should be given before a grant for $j_2$. However, in order to specify this in e.g., LTL (Linear Temporal Logic), we would have to explicitly write this statement for every pair of processes $j_1,j_2$. In the worst case, this could entail a blowup of $k!$ in the size of the formula, which incurs a further exponential blowup during model-checking algorithms.

This drawback, however, vanishes when we consider a \emph{symmetric} system: intuitively, a system is symmetric if permuting the input signals generates an output sequence of similarly permuted outputs. If a system satisfies this property, the it is enough to check whether it satisfies a representative specification. Indeed, any permutation of the processes is guaranteed to be equivalently satisfied.

Unfortunately, deterministic systems are unlikely to be completely symmetric, unless they are very naive (e.g., no grants are ever given). Indeed, tie-breaking in deterministic systems has an inherent asymmetry to it. In \emph{probabilistic} systems, however, no asymmetry is needed to break ties -- one can randomly choose a result. 

In this paper, we consider several notions of symmetry for probabilistic transducers, and their corresponding decision procedures. 
We start with the most restrictive version of symmetry, in which a transducer $\cT$ is symmetric under a permutation if the distribution of outputs that are generated for an input sequence $x$ is identical to the distribution of permuted outputs for the permuted input sequence (\cref{sec:symm}). We show that deciding whether a transducer is symmetric under a given permutation is decidable in polynomial time, and use basic results in group theory to give a similar result for deciding whether a transducer is symmetric under all permutations in a permutation group.

We then proceed to study approximate notions of symmetry, in order to capture cases where a system is not fully symmetric, but still may exhibit some symmetrical properties. On the negative side, using results on probabilistic automata, we show that an $L_\infty$ approximation variant of symmetry results in undecidability. On the positive side, we study two variants of symmetry that only take into account the Parikh image of the output signals, and we are able to use results on probabilistic automata with rewards to obtain efficient decidability of symmetry for these variants (\cref{sec:approx_sym}). 

Finally, we study a qualitative version of symmetry, which offers a coarse ``nondeterministic'' approximation of symmetry (\cref{sec:qual_sym}). We show that deciding whether a system is qualitatively symmetric is $\PSPACE$ complete.

The notion of symmetry is not only appealing for symmetry reductions in specification, but also as a standalone feature for the \emph{explainability} of model checking: standard model-checking algorithms can output a counterexample whenever a system does not satisfy its specification. This gives the designer insight as to what is wrong with either the system or the specification. On the other hand, when the result of model checking is that a system does satisfy its specification, no additional information is typically given. While this is ``good news'', a designer often wants some information as to ``why'' the system is correct. In particular, the designer may be concerned that the specifications were too easy to satisfy (e.g., in vacuous specifications \cite{ball2008vacuity}).
In this case, symmetry provides some information. Indeed, symmetry can be easily witnessed (as we show in~\cref{rmk:explainability}), so the designer can be convinced that any weakness of the specification, or any flaw of the system, is not biased toward a specific process, and will arise regardless of a specific order of processes. In addition, it shows that if the system satisfies e.g., liveness properties, then it satisfies them with the same ``good event intervals'' regardless of process identities.

\vspace*{-.3cm}
\paragraph*{Related work}
Process symmetry~\cite{clarke1996exploiting,ip1996better,emerson1996symmetry,lin2016regular} and more general symmetry reductions~\cite{sistla2000smc,spermann2008prob,wahl2010replication} have been studied since the 90's, typically in the context of alleviating the state-explosion problem. Symmetry can either be specified by the designer or user~[13,24,25], or detected automatically~[15,16,32]. 

A close approach to our work here is~\cite{lin2016regular}, where the problem of detecting process symmetries is studied. There, however, parametrized deterministic systems are studied, which shift the focus to the pattern of given symmetries (rather than our fixed-length permutations), and does not concern probabilities. 

Symmetry in the probabilistic setting was studied in~\cite{kwiatkowska2006symmetry,donaldson2005symmetry}, where model checking of probabilistic systems exploits known symmetries to avoid a state blowup by considering a quotient of the system under the symmetry. 

We remark that the works above typically focus on exact symmetries, and use them to reduce the state space, whereas the focus of this paper is to decide whether a symmetry exists, for various types of (not necessarily exact) symmetries, and to use the symmetry to avoid blowup in the specification, as well as to give the user insight regarding the correctness of the system.

Due to lack of space, some proofs appear in the appendix.

\vspace*{-10pt}
\section{Preliminaries}
\paragraph*{Probabilities and Distributions}
Consider a finite set $S$. A \emph{distribution} over $S$ is a function $\mu:S\to [0,1]$ such that $\sum_{s\in S}\mu(s)=1$. We denote the space of all distributions over $S$ by $\Delta(S)$. Given a distribution $\mu$, an \emph{event} is a subset\footnote{In general $E$ needs to be a \emph{measurable subset}, but since we only consider finite sets, any subset is measurable.} $E\subseteq S$, and its \emph{probability} under $\mu$ is $\Pr(E)=\sum_{e\in E}\mu(e)$.
For an element $s\in S$, the \emph{Dirac distribution} $\dirac{s}$ is given by $\dirac{s}(r)=\begin{cases}
1 & r=s,\\
0 & r\neq s.
\end{cases}$
The \emph{support} of a distribution $\mu$ is $\supp(\mu)=\{s\in S\ST \mu(s)>0\}$.

Given sets $S_1,\ldots,S_n$ and distributions $\mu_1,\ldots,\mu_n$ such that $\mu_i\in \Delta_i$ for every $1\le i\le n$, a natural \emph{product distribution}  $\mu$ is induced on the product space $S_1\times\cdots\times S_n$ where $\mu(s_1,\ldots,s_n)=\prod_{i=1}^{n}\mu_i(s_i)$.

\paragraph*{Probabilistic Transducers and Automata}
Consider two finite sets $I$ and $O$ of input and output signals, respectively. An $I/O$ \emph{probabilistic transducer} (henceforth just \emph{transducer}) is $\cT=\tup{I,O,S,s_0,\delta,\lab}$ where $S$ is a finite set of states, $s_0$ is an initial state, $\delta: S\times \tI\to \dist(S)$ is a transition function, assigning to each $($state,letter$)$ pair a distribution of successor states, and $\lab:S\to \tO$ is a labelling function. 

For a word $x=\vec{i}_1 \cdot \vec{i}_2 \cdots \vec{i}_n\in \Ip$, a \emph{run} of $\cT$ on $x$ is a sequence $\rho=q_0,q_1,\ldots,q_{n}$ where $q_0=s_0$, 
and the \emph{probability} of the run $\rho$ is $\prod_{j=0}^{n-1} \delta(q_j,\vec{i}_{j+1})(q_{j+1})$. Note that indeed this induces a probability measure $\mu$ on $\{s_0\}\times S^n$ via the product distribution.

A run $\rho$ is \emph{proper} if $\rho\in \supp(\mu)$. That is, if it has positive probability. We denote the space of proper runs by $\runs(\cT,x)$. In the following, we usually refer only to proper runs, and we omit the term ``proper'' when it is clear from context.
We extend the labelling function $\lab$ to runs by $\lab(\rho)=\lab(q_1)\cdot \lab(q_2)\cdots \lab(q_{n})$. Observe that we ignore the labelling of the initial state, and only consider nonempty words, to avoid edge cases.

For $x\in \Ip$ and $y\in \Op$ such that $|x|=|y|$, we denote by $\cT(x)=y$ the event $\{\rho\in \runs(\cT,x)\ST \lab(\rho)=y\}$. Thus, $\Pr(\cT(x)=y)$ is the probability that the output generated by $\cT$ on input $x$ is exactly $y$. We denote by $x\otimes y\in \IOo$ the combined word $(\vec{i}_1 \cup \vec{o}_1)\cdot (\vec{i}_2 \cup \vec{o}_2)\cdots (\vec{i}_n \cup \vec{o}_n)$.

The sets $I$ and $O$ are called \emph{corresponding signals} if $I=\{i_1,\ldots, i_k\}$ and $O=\{o_1,\ldots,o_k\}$. Intuitively, for $1\le j\le k$ we think of $i_j$ as a request generated by a process $j$, and of $o_j$ as a corresponding grant generated by the system. 

A \emph{probabilistic automaton (PA)} is $\cA=\tup{Q,\Sigma,\delta,q_0,F}$ where $Q$ is a finite set of states, $\Sigma$ is a finite alphabet, $\delta:Q\times \Sigma\to \dist(Q)$ is a probabilistic transition function, $q_0\in Q$ is an initial state, and $F\subseteq Q$ is a set of accepting states. Similarly to transducers, an input word $x\in \Sigma^*$ induces a probability measure on the set $\runs(\cA,x)$ of runs of $\cA$ on $x$. Then, we denote by $\cA(x)$ the probability that a run of $\cA$ on $x$ is accepted, i.e. ends in a state in $F$.

\paragraph*{Permutations}
We assume familiarity with basic notions in group theory (see e.g.~\cite{cameron1999permutation}).
A \emph{permutation} of the set $[k]=\{1,\ldots,k\}$ is a bijection $\pi:[k]\to [k]$. A standard representation of permutations is by a \emph{cycle decomposition}, where, for example, the cycle $(1\ 2\ 7)$ represents the permutation $\pi$ where $\pi(1)=2,\pi(2)=7, \pi(7)=1$, and for all other elements we have $\pi(j)=j$.  The set of all permutations on $[k]$, equipped with the functional composition operator $\circ$ forms the \emph{symmetric group} $\cS_k$. Any subgroup of $\cS_k$ is referred to as a \emph{permutation group}. 
A \emph{generating set} of a permutation group $G$ is a finite set $X=\{\pi_1,\ldots,\pi_m\}$ such that every permutation $\tau\in G$ can be expressed as a composition of the elements in $X$. For such a set $X$, we denote the group generated by it by $\tup{X}$.
It is well known that $\{(1\ 2),(1\ 2\ \ldots\ k)\}$ is a generating set of $\cS_k$ (\cite{cameron1999permutation}).

Consider corresponding signals $I=\{i_1,\ldots, i_k\}$ and $O=\{o_1,\ldots,o_k\}$, and let $\pi\in \cS_k$. For a letter $\vec{i}=\{i_{j_1},\ldots,i_{j_m}\}\in \tI$, we define $\pi(\vec{i})=\{i_{\pi{j_1},\ldots,i_{\pi(j_m)}}\}$. That is, $\pi$ permutes the signals given in $\vec{i}$.\footnote{Formally, we would actually need $I$ to be an ordered set. However, the order will be implied by the naming convention, so we let $I$ be a set.} Then, for a word $x=\vec{i}_1\cdot \vec{i}_2\cdots \vec{i}_n\in \Ip$, we define $\pi(x)=\pi(\vec{i}_1)\cdot \pi(\vec{i}_2)\cdots \pi(\vec{i}_n)$. Similar definitions hold for $O$. 
Unless explicitly stated otherwise, we henceforth assume $I$ and $O$ are corresponding signals.

\section{Symmetric Probabilistic Transducers}
\label{sec:symm}
Let $\cT=\tup{I,O,S,s_0,\delta,\lab}$ be an $I/O$ transducer over $I=\{i_1,\ldots, i_k\}$ and $O=\{o_1,\ldots,o_k\}$, and let $\pi\in \cS_k$. We say that $\cT$ is $\pi$-symmetric if for every $x\in \Ip$ and $y\in \Op$ it holds that $\Pr(\cT(x)=y)=\Pr(\cT(\pi(x))=\pi(y))$. That is, $\cT$ is $\pi$-symmetric if whenever we permute the input by $\pi$, the resulting distribution on outputs is permuted by $\pi$ as well.

\begin{example}
	\label{xmp:robin}
	Consider a Round-Robin arbiter over three processes, as depicted in~\cref{fig:fig3}. At each state, the arbiter looks for a request from a single processor $j$, and grants it if it is on, then moves to a state corresponding to process $j+1 \pmod 3$. Observe that this is a deterministic transducer, except that the initial state is unspecified.
	
	\begin{figure}[ht]
		\centering
		\includegraphics[trim=44 54 0 37,clip]{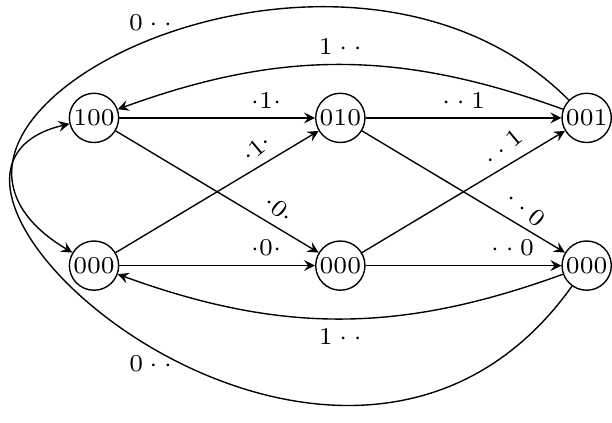}
		\caption{A transducer for a Round Robin arbiter. The labels on the transitions and states are the characteristic vectors of the labels, with $\cdot$ as placeholders. Thus, e.g., $100$ is $\{i_1\}$, and $\cdot \cdot 1$ is any $\vec{i}$ such that $i_3\in \vec{i}$.
			The initial state is unspecified, see \cref{xmp:robin}.}
		\label{fig:fig3}
	\end{figure}
	
	Consider the case where we let the state marked $001$ be initial, which corresponds to letting the first process start. In this case, the transducer is not $\pi$-symmetric for $\pi=(1\ 2\ 3)$. Indeed, the input word $100$ will generate output $100$, but its permutation $\pi(100)=010$ generates output $000\neq \pi(100)$.
	
	However, if we introduce a probabilistic initial state, that chooses each state of $100,010,001$ as the next state, each with probability $\frac13$, the transducer becomes $\pi$-symmetric for any $\pi\in \cS_3$.
	\qed
\end{example}

Consider a permutation group $G=\tup{X}$ generated by $X=\{\pi_1,\ldots,\pi_m\}$. We say that $\cT$ is $G$-symmetric if it is $\pi$-symmetric for every $\pi\in G$.
Toward understanding symmetry, we start by showing that it is enough to consider symmetry under the generators.
\begin{lemma}
	\label{lem:symmetry_composition}
	Consider an $I/O$ transducer $\cT$ over $I=\{i_1,\ldots, i_k\}$ and $O=\{o_1,\ldots,o_k\}$. If $\cT$ is $\pi$-symmetric and $\tau$-symmetric for $\pi,\tau\in \cS_k$, then $\cT$ is $\pi\circ \tau$-symmetric.
	\end{lemma}
An immediate corollary of \cref{lem:symmetry_composition} is that in order to check whether $\cT$ is $G$-symmetric, it suffices to check whether it is symmetric with respect to the generators of $G$.

\begin{corollary}
	\label{cor:symmetry_group_iff_generator}
	Consider an $I/O$ transducer $\cT$ and a permutation group $G$ with generators $X$, then $\cT$ is $G$-symmetric iff it is $\pi$-symmetric for every $\pi\in X$.
	\end{corollary}

\begin{remark}[Symmetry for Explainability]
	\label{rmk:explainability}
\cref{cor:symmetry_group_iff_generator} is key to using symmetry for explainability of model checking. Indeed, it shows that we can convince a designer that a system is e.g., $\S_k$-symmetric by showing that it is symmetric under the two generators. That is, the witness for symmetry consists of demonstrating symmetry on two permutations. As discussed in~\cref{sec:intro}, once the designer is convinced the system possesses symmetric properties, she gains some insight to the possible reasons that make the system correct, or to possible behaviour of bugs, when the system is incorrect.
\qed
\end{remark}

The fundamental problem about symmetry of probabilistic transducers is whether a transducer is $\pi$-symmetric for a given permutation $\pi$. We now show that this problem can be solved in polynomial time.
\begin{theorem}
	\label{thm:deciding_symmetry_permutation_ptime}
	The problem of deciding, given an  $I/O$ transducer $\cT$ and a permutation $\pi\in \cS_k$, whether $\cT$ is $\pi$-symmetric, is solvable in polynomial time.
\end{theorem}
\begin{proof}
	Given two probabilistic automata $\cA$ and $\cB$ over the alphabet $\Sigma$, the problem of determining whether $\cA(x)=\cB(x)$ for every $x\in \Sigma^*$, dubbed the \emph{equivalence problem}, is solvable in polynomial time~\cite{gimbert2010probabilistic,schutzenberger1961definition,tzeng1992polynomial}. Our proof is by reduction of the problem at hand to the equivalence problem for probabilistic automata.
	
	Consider an $I/O$ transducer $\cT=\tup{I,O,S,s_0,\delta,\lab}$ over $I=\{i_1,\ldots, i_k\}$ and $O=\{o_1,\ldots,o_k\}$, and let $\pi\in \cS_k$. We construct from $\cT$ two PAs $\cA$ and $\cB$.  Intuitively, $\cA$ mimics the behaviour of $\cT$, by reading words over $\tIO$, and accepting a word $w\in \IOp$ with probability $\mu$ iff $\cT$, when reading the inputs that appear in $w$, generates the outputs that appear in $w$ with probability $\mu$. The PA $\cB$ works exactly like $\cA$, but permutes both the inputs and outputs by $\pi$.

	Formally, $\cA=\tup{S\cup\{q_{\bot}\},\tIO,\eta,s_0,S}$ and $\cB=\tup{S\cup\{q_{\bot}\},\tIO,\zeta,s_0,S}$ where $q_{\bot}$ is a new state, and the transition functions are defined as follows. Let $q\in S$ and $\sigma=\vec{i}\cup\vec{o}$ with $\vec{i}\in \tI$ and $\vec{o}\in \tO$, and let $V_p=\sum_{
		p\in S,\ 
		\lab(p)=\vec{o}} \delta(q,\vec{i})(p)$ be the probability assigned by $\cT$ to seeing a state labelled $\vec{o}$ after reading $\vec{i}$ in state $q$, then $\eta(q,\sigma)\in \dist(S\cup \{q_\bot\})$ is the following distribution: 
	\[
		\eta(q,\sigma)(p)=\begin{cases}
			\delta(q,\vec{i})(p) & \mbox{ if } p\in S \mbox{ and } \lab(p)=\vec{o}\\
			0 	& \mbox{ if } p\in S \mbox{ and } \lab(p)\neq\vec{o}\\
			1-V_p & \mbox{ if } p=q_{\bot} 
		\end{cases}
	\]
	In addition, $\eta(q_{\bot},\sigma)(q_{\bot})=1$ (so $q_{\bot}$ is a rejecting sink). We demonstrate the construction of $\cA$ in \cref{fig:tran2pa0,fig:tran2pa1}.
	
	\begin{figure}[ht]
		\centering
		\subcaptionbox{Transition in $\cT$\label{fig:tran2pa0}}
		{\includegraphics[scale=.8]{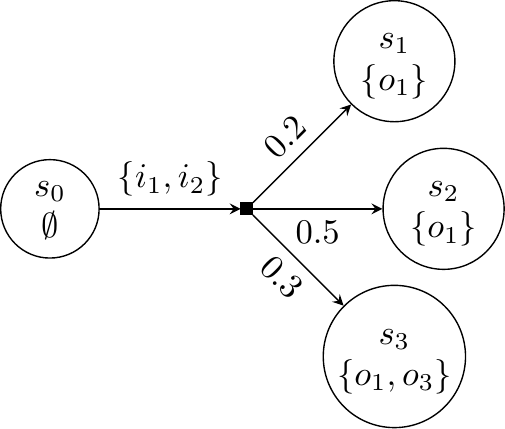}}\quad
		\subcaptionbox{Transition in $\cA$\label{fig:tran2pa1}}
		{\includegraphics[scale=1]{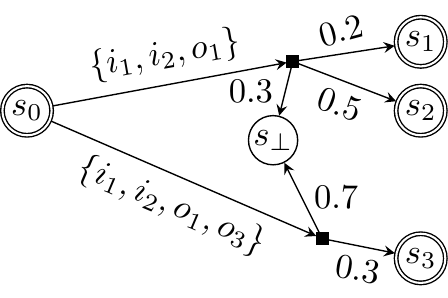}}\quad
		\subcaptionbox{Transition in $\cB$\label{fig:tran2pa2}}
		{\includegraphics[scale=1]{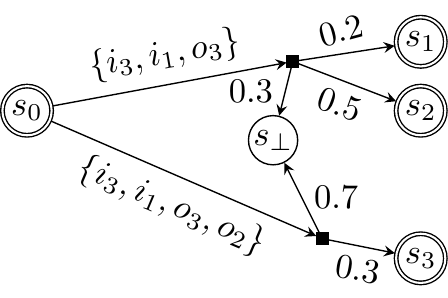}}
		\caption{A transition in a transducer $\cT$ over $I=\{i_1,i_2,i_3\}$ and $O=\{o_1,o_2,o_3\}$, and the corresponding transitions in $\cA$ and $\cB$, under the permutation $\pi=(1\ 2\ 3)$.
		Observe that the transition in $\cB$ corresponds to the inverse permutation, $\pi^{-1}=(3\ 2\ 1)$, so that e.g., $\pi(\{i_3,i_1\})=\{i_1,i_2\}$.
		}\label{fig:tran2pa}
	\end{figure}

	The construction of $\cB$ is similar, but accounts for the permutation $\pi$. Let $q\in S$ and $\sigma=\vec{i}\cup\vec{o}$ with $\vec{i}\in \tI$ and $\vec{o}\in \tO$, and let $U_p=\sum_{
		p\in S,\ 
		\lab(p)=\pi(\vec{o})} \delta(q,\pi(\vec{i}))(p)$ be the probability assigned by $\cT$ to seeing a state labelled $\pi(\vec{o})$ after reading $\pi(\vec{i})$ in state $q$, then $\zeta(q,\sigma)\in \dist(S\cup \{q_\bot\})$ is the following distribution: 
	\[
	\zeta(q,\sigma)(p)=\begin{cases}
	\delta(q,\pi(\vec{i}))(p) & \mbox{ if } p\in S \mbox{ and } \lab(p)=\pi(\vec{o})\\
	0 	& \mbox{ if } p\in S \mbox{ and } \lab(p)\neq\pi(\vec{o})\\
	1-U_p & \mbox{ if } p=q_{\bot} 
	\end{cases}
	\]
	In addition, $\zeta(q_{\bot},\sigma)(q_{\bot})=1$ (so $q_{\bot}$ is a rejecting sink). We demonstrate the construction of $\cB$ in \cref{fig:tran2pa0,fig:tran2pa2}.
	
	Consider words $x\in \Ip$ and $y\in \Op$. Since $q_{\bot}$ is the only rejecting state in both $\cA$ and $\cB$, then by construction it is easy to see that $\cA(x\otimes y)=\Pr(\cT(x)=y)$ and $\cB(x\otimes y)=\Pr(\cT(\pi(x))=\pi(y))$. Thus, we have that $\cA$ and $\cB$ are equivalent iff $\cT$ is $\pi$-symmetric, and since equivalence can be decided in polynomial time, we are done.
\end{proof}

Combining \cref{thm:deciding_symmetry_permutation_ptime} with \cref{cor:symmetry_group_iff_generator}, we have the following.
\begin{corollary}
	\label{cor:deciding_symmetry_group_ptime}
	The problem of deciding, given an  $I/O$ transducer $\cT$ and a finite set of generators $X=\{\pi_1,\ldots,\pi_m\}$, whether $\cT$ is $\tup{X}$-symmetric, is solvable in polynomial time.
\end{corollary}

In particular, since the symmetric group $\cS_k$ is generated by two permutations $\{(1\ 2),(1\ 2\ \ldots\ k)\}$, we have the following.
\begin{corollary}
	\label{cor:deciding_S_k_ptime}
	The problem of deciding, given an $I/O$ transducer $\cT$, whether $\cT$ is $\cS_k$-symmetric, is solvable in polynomial time.
\end{corollary}

\section{Approximate Symmetry}
\label{sec:approx_sym}
While aspiring to obtain symmetric systems is noble, in practice exact symmetry may be too strong a requirement, for example if the source of randomness supplies binary bits, and one needs e.g., $\frac13$ probability, then only an approximate probability can be used.
Thus, it is reasonable to seek approximate notions of symmetry.

\subsection{$L_\infty$ Symmetry}
\label{sec:L_inf_approx}
The most straightforward approach toward approximate symmetry in probabilistic transducers is induced by the the $L_\infty$ norm, as follows.
Let $\cT$ be an $I/O$-transducer, let $\pi\in \cS_k$, and let $\epsilon>0$. We say that $\cT$ is $(\epsilon,\pi)$-symmetric if $|\Pr(\cT(x)=y)-\Pr(\cT(\pi(x))=\pi(y))|\le \epsilon$ for every $x\in \Ip$ and for every $y\in \Op$. 
That is, permuting the inputs by $\pi$ perturbs the output distribution by at most $\epsilon$.

Unfortunately, as we now show, approximate symmetry is undecidable.
\begin{theorem}
	\label{thm:deciding_approx_sym_undecidable}
	The problem of deciding, given an  $I/O$ transducer $\cT$ a permutation $\pi\in \cS_k$ and $\epsilon>0$, whether $\cT$ is $(\epsilon,\pi)$-symmetric, is undecidable.
\end{theorem}
\begin{proof}
	The \emph{emptiness problem} for PA is to decide, given a PA $\cA$ over $\Sigma$ and a threshold $\lambda\in [0,1]$, whether there exists a word $w\in \Sigma^*$ such that $\cA(w)> \lambda$. This problem is known to be undecidable~\cite{paz2014introduction,madani2003undecidability,gimbert2010probabilistic}. 
	
	We show that approximate symmetry is undecidable via a reduction from (the complement of) a restriction of the emptiness problem, where the given PA is over the alphabet $\{0,1\}$. 
	The problem remains undecidable under this restriction, as we can encode any larger alphabet $\Gamma$ using fixed-length sequences in $\{0,1\}^{d}$, such that while reading the $d$ symbols that compose a single letter in $\Gamma$, the states are not accepting (and hence we do not introduce a word whose acceptance probability is above $\lambda$). 

	We start with an intuitive description of the reduction, depicted in~\cref{fig:reductionPA}.
	\begin{figure}[ht]
		\centering
		\includegraphics[scale=.8]{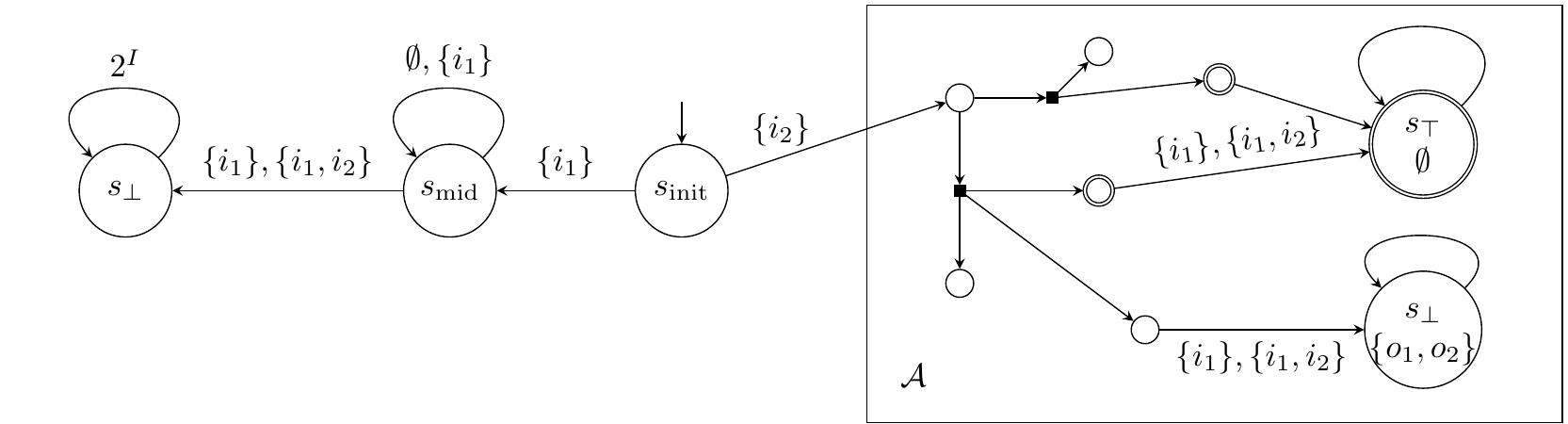}
		\caption{The transducer constructed from a PA. The black squares denote probabilistic branching.}
		\label{fig:reductionPA}
	\end{figure}

	 Consider a PA $\cA$ over the alphabet $\Sigma=\{0,1\}$. We construct a transducer $\cT$ over $I=\{i_1,i_2\}$ and $O=\{o_1,o_2\}$ which has two components. Initially, if $\cT$ sees the input $\{i_2\}$, it moves to a component which mimics $\cA$ using the alphabet $\{\emptyset,\{i_2\}\}$ instead of $\{0,1\}$. At this stage, all the states are marked with the output $\{o_1,o_2\}$. If at any point the input signal $i_1$ is given, i.e. the letter $\{i_1\}$ or $\{i_1,i_2\}$, then $\cT$ proceeds to a state labelled $\{o_1,o_2\}$ from non-accepting states of $\cA$, and to a state labelled $\emptyset$ from accepting states. Thus, a word of the form $\{i_2\}\cdot x\cdot \{\{i_1\}, \{i_1,i_2\}\}^*$ with $x\in \{\emptyset,\{i_2\}\}^n$ would yield an output of the form $\emptyset^{n+1}\cdot \emptyset^*$ with probability $\cA(x)$ and of the form $\emptyset^{n+1}\cdot \{o_1,o_2\}^*$ with probability $1-\cA(x)$. Observe that both output possibilities are invariant under the permutation $(1\ 2)$.
	
	If, initially, $\cT$ sees the input $\{i_1\}$, it moves to a state labelled $\emptyset$, which loops as long as $\{i_1\}$ or $\emptyset$ are seen. Then, if $\{i_2\}$ or $\{i_1,i_2\}$ is seen, it moves to a sink labelled $\{o_1,o_2\}$. Essentially, this component mimics the output sequence of a rejecting run of $\cA$ in the first component, under the permutation $(1\ 2)$. Hence, taking $\epsilon=\lambda$, we have that $\cT$ is $(\epsilon, (1\ 2))$- symmetric iff there does not exist a word $x$ such that $\cA(x)> \lambda$.

	We proceed to give the precise reduction.
	Consider a PA $\cA=\tup{Q,\Sigma,\delta,q_0,F}$ with $\Sigma=\{0,1\}$, we construct an $I/O$ transducer $\cT=\tup{I,O,S,\sinit,\eta,\lab}$ as follows. The states of $\cT$ are $S=Q\cup \{\smid,\sinit,s_\top, s_{\bot}\}$, where $s_\bot\notin Q$, and the input and output sets are $I=\{i_1,i_2\}$ and $O=\{o_1,o_2\}$. The labelling function is given by $\lab(q)=\emptyset$ for all $q\in Q$,  $\lab(s_{\bot})=O=\{o_1,o_2\}$, and $\lab(\sinit)=\lab(\smid)=\{\emptyset\}$. The transition function, as depicted in \cref{fig:reductionPA}, is defined as follows.
	
	First, for every $q\in Q$ and $\vec{i}\in \{\emptyset,\{i_2\}\}$, we have $\eta(q,\vec{i})=\delta(q,\vec{i})$, where we identify $\{\emptyset,\{i_2\}\}$ with $\{0,1\}$ in an arbitrary bijective manner. Next, if $q\in F$, then $\eta(q,\{i_1\})=\eta(q,\{i_1,i_2\})=\dirac{s_\top}$, and if $q\notin F$ then $\eta(q,\{i_1\})=\eta(q,\{i_1,i_2\})=\dirac{s_\bot}$. 
	The remaining transitions are\\
	\begin{tabular}{ l l }
	$\eta(\sinit,\{i_1\})=\dirac{\smid}$,	& $\eta(\smid,\emptyset)=\eta(\smid,\{i_1\})=\dirac{\smid}$,  \\ 
	 $\eta(\sinit,\{i_2\})=\dirac{q_0}$,	& $\eta(\smid,\{i_2\})=\eta(\smid,\{i_1,i_2\})=\dirac{s_\bot}$,\\ 
	$\eta(\sinit,\emptyset)=\eta(\sinit,\{i_1,i_2\})=\dirac{s_\bot}$,	& \\ 
	\end{tabular} \\
	and for every $\vec{i}\in\tI$ we have $\eta(s_\bot,\vec{i})=\dirac{s_\bot}$ and $\eta(s_\top,\vec{i})=\dirac{s_\top}$.  
	
	Let $\pi=(1\ 2)$ and $\epsilon=\lambda$. Keeping our identification of $\{\emptyset,\{i_2\}\}$ with $\{0,1\}$, we claim that there exists a word $x'\in \{\emptyset,\{i_2\}\}^*$ such that $\cA(x')>\lambda$ iff there exists words $x\in \Ip$ and $y\in \Op$ such that $|\Pr(\cT(x)=y)-\Pr(\cT(\pi(x))=\pi(y))|>\epsilon$ (i.e. $\cT$ is not $(\epsilon,\pi)$-symmetric). 
	Observe that $\lab$ assigns only the labels $\emptyset$ and $\{o_1,o_2\}$, both of which are invariant under $\pi$. Thus, the latter condition becomes 
	\begin{equation}
	\label{eq:undec_eq}
	|\Pr(\cT(x)=y)-\Pr(\cT(\pi(x))=y)|>\epsilon.
	\end{equation}
	
	The correctness proof can be found in~\cref{apx:correctness}.

\end{proof}

A-priori, the fact that $(\epsilon,\pi)$-symmetry is undecidable does not mean that approximate symmetry for an entire permutation group is undecidable, not that for fixed $\epsilon$ the problem is undecidable. Unfortunately, however, the proof of \cref{thm:deciding_approx_sym_undecidable} uses the permutation group $\cS_2$, whose only nontrivial permutation is $(1\ 2)$. Moreover, the reduction uses the given threshold $\lambda$ as is, by setting $\lambda=\epsilon$, and the emptiness problem is known to be undecidable even when $\lambda$ is a fixed number in $(0,1)$. Thus, we have the following.
\begin{corollary}
	\label{cor:approximate_undecidable_fixed}
	For every $\epsilon\in (0,1)$, the problem of deciding, given an $I/O$ transducer $\cT$ whether $\cT$ is $(\epsilon,\pi)$-symmetric for every $\pi\in \cS_k$, is undecidable.
	\end{corollary}

\begin{remark}[Composability]
	\label{rmk:composability_approx}
	While undecidability of $(\epsilon,\pi)$-symmetry is unfortunate, the reader may take solace in the fact that $(\epsilon,\pi)$-symmetry is anyway not preserved under composition. Indeed, if $\cT$ is $(\epsilon,\pi)$-symmetric and $(\delta,\tau)$-symmetric, it only guarantees that it is $(\delta+\epsilon,\tau\cdot \pi)$-symmetric. Thus, in order to ensure symmetry over a group, a sound method would have to take into account the \emph{diameter} of the group. This, however, may lose completeness. Thus, $(\epsilon,\pi)$-symmetry is not a robust notion.
	\end{remark}

\subsection{Parikh Symmetry}
\label{sec:parikh}
The notions of symmetry studied so far have a ``letter-by-letter'' flavour, where we compare the distribution of specific outputs for a given inputs. We now turn to study a different notion of symmetry, that abstracts away the order of the output symbols, and draws instead on the Parikh image of the computation.

Let $I=\{i_1,\ldots, i_k\}$ and $O=\{o_1,\ldots,o_k\}$. For a word $y=\vec{o}_1\cdots \vec{o}_n\in \tO$, and $1\le j\le k$, define $\num(y,j)=|\{m\ST o_j\in \vec{o}_m\}|$ to be the number of occurrences of $o_j$ in $y$. Then, we define the \emph{Parikh image}\footnote{Observe that this is not the standard Parikh image, in that it is the image with respect to signals in $O$, rather than to letters in $2^O$.} of $y$ to be $\park(y)=(\num(y,1),\ldots,\num(y,k))\in \bbN^k$.

Given a permutation $\pi$ and a vector $\vec{a}=(a_1,\ldots,a_k)\in \bbN^k$, we define $\pi(\vec{a})=(a_{\pi^{-1}(1)},\ldots,a_{\pi^{-1}(k)})$. Note that we use $\pi^{-1}$ so that the following relation holds: if e.g., $\pi(1)=3$, then index $3$ in $\pi(\vec{a})$ contains $a_1$.  

Consider an $I/O$ transducer $\cT$ and a word $x\in \Ip$. The outputs of $\cT$ on $x$ induce a probability measure on (a finite subset of) $\bbN^k$, where for a vector $\vec{a}\in \bbN^k$ we have $\Pr(\cT(x)=\vec{a})=\sum_{y: \park(y)=\vec{a}}\Pr(\cT(x)=y)$. 
We can thus also consider the \emph{expected} value of the Parikh image, given by $\bbE[\park(\cT(x))]=\sum_{y} \Pr(\cT(x)=y) \park(y)$ (where the product is element-wise, so this is a vector in $\bbN^k$).

Parikh images give rise to two measures of symmetry: given a permutation $\pi$, we say that $\cT$ is \emph{$\pi$-Parikh distribution symmetric} if for every $x\in \Ip$ and every $\vec{a}\in \bbN^k$ we have $\Pr(\park(\cT(x))=\vec{a})=\Pr(\park(\cT(\pi(x)))=\pi(a))$.  That is, every word $x$ induces the same distribution of Parikh images as $\pi(x)$ does for the permuted images. 
A weaker notion of symmetry uses expectation: we say that $\cT$ is \emph{$\pi$-Parikh expected symmetric} if for every $x\in \Ip$ we have $\bbE[\park(\cT(x))]=\pi(\bbE[\park(\cT(\pi(x)))])$

Note that Parikh-symmetry assumes the number of occurrences of a certain output signal is meaningful. This is relevant when the output signals measure e.g., number of grants for requests, but makes less sense when the outputs represent e.g., a choice between channels through which a message is routed.

Our algorithmic results about Parikh symmetry use a translation to \emph{probabilistic reward automata} (PRA)~\cite[Section 5]{kiefer2014stability}. A PRA is a PA $\cA=\tup{Q,\Sigma,\delta,q_0,F}$ equipped with a \emph{reward function} $\rew:Q\to \{0,1\}^k$ for some $k\in \bbN$.\footnote{The rewards in~\cite{kiefer2014stability} also allow $-1$ rewards, and is set on the transitions of the PRA. Since it is trivial to push rewards from the states to the transitions, our model is simpler.} The rewards are summed along a run, and the value of a word $w\in \Sigma^*$, denoted $\rew(w)$, is the expected reward, that is, the weighted sum of the rewards along all runs, weighted by their respective probabilities. We denote by $\cA(w)$ the distribution of reward vectors in $\bbN^k$, induced by the runs of $\cA$ on $w$.

In order to reason about Parikh images, we propose the following translation.
\begin{lemma}
	\label{lem:tran_to_PRA}
	Given an $I/O$ trandsucer $\cT$, we can construct two PRAs $\cA,\cB$ over the alphabet $\tI$ and with reward function of dimension $k=|I|$, such that for every $x\in \Ip$ and for every $\vec{a}\in \bbN^k$, we have that $\Pr(\cA(w)=\vec{a})=\Pr(\park(\cT(x))=\vec{a})$, and $\Pr(\cB(w)=\vec{a})=\Pr(\park(\cT(\pi(x)))=\pi(\vec{a}))$.
	\end{lemma}
%
%

In \cite{kiefer2014stability}, the problems of distribution-equivalence and expected-equivalence are solves, with complexities $\NC$ and $\RNC$, respectively, where $\NC$ is the class of problems solvable using circuits of polynomial size and polylogarithmic depth, and $\RNC$ is its randomized analogue. It is known that $\NC\subseteq \P$ and $\RNC\subseteq \RP$.

The distribution-equivalence and expected-equivalence problems, applied to the automata $\cA$ and $\cB$ obtained as per \cref{lem:tran_to_PRA}, exactly correspond to $\pi$-distribution symmetry and $\pi$-expected symmetry of $\cT$, respectively. We thus have the following.
\begin{theorem}
	\label{thm:parikh_complexity}
	The problem of deciding, given an $I/O$ transducer $\cT$ and a permutation $\pi$, whether it is $\pi$-Parikh distribution symmetric (resp. $\pi$-Parikh expected symmetric), is in $\NC$ (resp. $\RNC$).
\end{theorem}
Both notions of Parikh symmetry can be easily shown respect composition, analogously to \cref{lem:symmetry_composition}, in that if $\cT$ is both $\pi$- and $\tau$- Parikh distribution/expected symmetric, then it is also $\pi\circ \tau$-Parikh distribution/expected symmetric. Thus, we conclude this section with the following.
\begin{theorem}
	\label{thm:parikh_complexity_group}
	The problem of deciding, given an  $I/O$ transducer $\cT$ and a finite set of generators $X=\{\pi_1,\ldots,\pi_m\}$, whether it is $\pi$-Parikh distribution symmetric (resp. $\pi$-Parikh expected symmetric) for every $\pi\in \tup{X}$, is in $\NC$ (resp. $\RNC$).
\end{theorem}

\section{Qualitative Symmetry}
\label{sec:qual_sym}
\cref{sec:L_inf_approx} rules out a decidable quantitative approximation for symmetry that takes into account the order of the input (at least in the sense of \cref{thm:deciding_approx_sym_undecidable}). In lieu of such an approximation, 
we turn to study a qualitative approximation, whereby we only require that permuting the input does not alter the support of the output distribution.

Let $\cT$ be an $I/O$ transducer, and let $\pi\in \cS_k$. We say that $\cT$ is \emph{$\pi$-qualitative-symmetric} if for every $x\in \Ip$ and $y\in \Op$ we have that $\Pr(\cT(x)=y)>0$ iff $\Pr(\cT(\pi(x))=\pi(y))>0$.

Observe that for every $x$ and $y$ as above, $\Pr(\cT(x)=y)>0$ iff there exists a run of $\cT$ on $x$ that is labelled $y$. Thus, in order to study qualitative symmetry, we can ignore the concrete probabilities in $\cT$, and only keep information on whether they are positive or not. Therefore, we essentially consider a nondeterministic transducer.

Using a similar translation to that in \ref{thm:deciding_symmetry_permutation_ptime}, but to NFAs instead of PAs, we have the following.
\begin{lemma}
	\label{lem:qual_PSPACE}
	The problem of deciding, given an $I/O$ transducer $\cT$ and a permutation $\pi$, whether $\cT$ is $\pi$-qualitative-symmetric, is in \PSPACE.
	\end{lemma}
%
%

We proceed to show a matching lower bound.
\begin{lemma}
	\label{lem:qual_PSPACE_hard}
	The problem of deciding, given an $I/O$ transducer $\cT$ and a permutation $\pi$, whether $\cT$ is $\pi$-qualitative-symmetric, is \PSPACE-hard.
\end{lemma}
\begin{proof}
	We show the problem is \PSPACE-hard via a reduction from the universality problem for NFAs over alphabet $\Sigma=\{0,1\}$ whose states are all accepting. That is, the problem of deciding, given an NFA $\cA=\tup{Q,\{0,1\},\delta,q_0,Q}$ (where all states are accepting), whether $L(\cA)=\Sigma^*$. This problem was shown to be \PSPACE-hard in \cite{kao2009nfas}.
	
	The reduction has a similar flavour as that of \cref{thm:deciding_approx_sym_undecidable}, in that we use the permutation to switch between components of the transducer. The components themselves, however, are somewhat different. 
	
	Let $\cA=\tup{Q,\{0,1\},\delta,q_0,Q}$ be an NFA over $\{0,1\}$ with all states accepting. 
	We construct a transducer $\cT=\tup{I,O,S,s_0,\eta,\lab}$ over $I=\{i_1,i_2\}$ and $O=\{o_1,o_2\}$ as follows. The states are $S=Q\cup \{\sinit,\smid,s_{\bot}\}$, with the labelling $\lab(q)=\emptyset$ for every $q\in Q$, $\lab(\sinit)=\lab(\smid)=\emptyset$, and $\lab(s_\bot)=\{o_1,o_2\}$. For simplicity, we treat the transition function as nondeterministic $\eta:S\times \tIO\to 2^S$. Technically, this can be thought of as specifying the support of the transition function, with arbitrarily chosen probabilities (e.g., uniform). Note, however, that we do not allow $\emptyset$ in the image of $\delta$, since we must be able to specify probabilities for the transitions.	
	Now, for every $q\in Q$ and $\vec{i}\in \tI$, and we define 
	\[
		\eta(q,\vec{i})=\begin{cases}
			\delta(q,0)\cup\{s_\bot\} & \text{ if }\vec{i}=\emptyset\\ 
			\delta(q,1)\cup\{s_\bot\} & \text{ if }\vec{i}=\{i_1,i_2\}\\ 
			\{q_\bot\} & \text{ otherwise}
		\end{cases}
	\]
	That is, within the $Q$ component, we identify $\Sigma=\{0,1\}$ with $\{\emptyset,\{i_1,i_2\}\}$, and whenever there are no corresponding transitions in $\cA$, or an ``invalid'' letter is seen, a transition is taken to $s_{\bot}$. Note that we add transitions to $s_\bot$ even when there are transition in $\cA$, which will play a role later on.
	The remaining transitions are as follows (see~\cref{fig:reductionNFA}).\\
	\begin{tabular}{ l l }
		$\eta(\sinit,\{i_1\})=\{q_0\}$, & $\eta(\sinit,\{i_2\})=\{\smid\}$, \\ $\eta(\sinit,\emptyset)=\eta(\sinit,\{i_1,i_2\})=\{s_\bot\}$, &  $\eta(\smid,\emptyset)=\eta(\smid,\{i_1,i_2\})=\{\smid,s_\bot\}$,\\ $\eta(\smid,\{i_1\})=\eta(\smid,\{i_2\})=\{s_\bot\}$, & and $\eta(s_\bot,\sigma)=\{s_\bot\}$.
	\end{tabular} 
	\begin{figure}[ht]
		\centering
		\includegraphics[scale=.8]{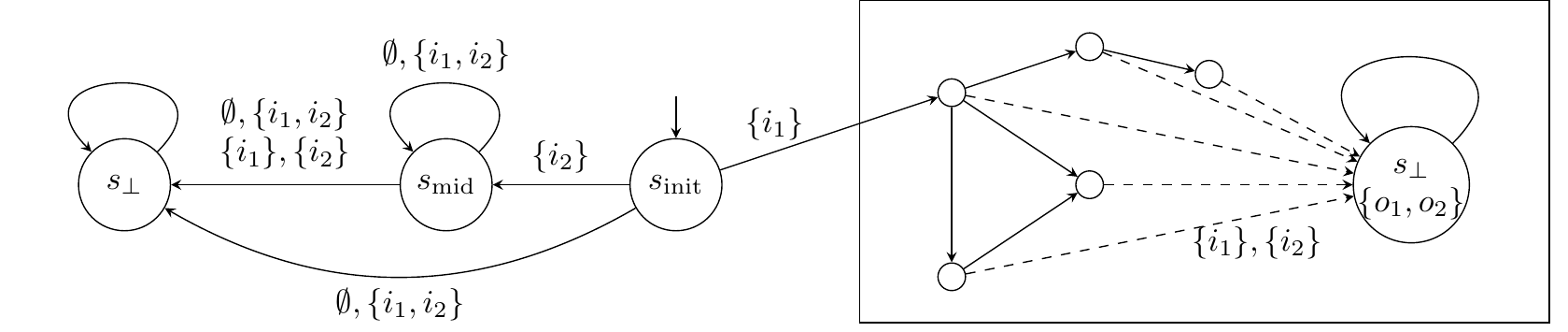}
		\caption{The transducer constructed from an NFA.}
		\label{fig:reductionNFA}
	\end{figure}

	\noindent Let $\pi=(1\ 2)$. We claim that $L(\cA)=\Sigma^*$ iff $\cT$ is $(1\ 2)$-qualitative-symmetric. 
	
	For the first direction, we prove the contrapositive. Assume $L(\cA)\neq \Sigma^*$, and let $w\in \Sigma^*\setminus L(\cA)$. Keeping our identification of $\Sigma=\{0,1\}$ with $\{\emptyset,\{i_1,i_2\}\}$, consider the word $x=\{i_1\}\cdot w$. Since there are no runs of $\cA$ on $w$, it follows that within the $Q$ component, after reading $w$, the only reachable state is $s_\bot$. Thus, if $z\in \Op$ is such that $\Pr(\cT(x)=z)>0$, then $z$ is of the form $\emptyset^+\cdot \{o_1,o_2\}^+$. In particular, let $y=\emptyset^{|w|+1}$, then $\Pr(\cT(x)=y)=0$. However, a possible run of $\cT$ on $\pi(x)$ is $\sinit,\smid^{|w|}$, which induces the labels $y=\pi(y)$. Thus, $\Pr(\cT(\pi(x))=\pi(y))>0$, so $\cT$ is not $\pi$-qualitative-symmetric.
	
	Conversely, assume that $L(\cA)=\Sigma^*$, and consider $x\in \Ip$ and $y\in \Op$. We claim that $\Pr(\cT(x)=y)>0$ iff $\Pr(\cT(\pi(x))=\pi(y))>0$. Observe that similarly to \cref{thm:deciding_approx_sym_undecidable}, all the labels on $\cT$ are invariant under $\pi$, so the above can be stated as 
	
	\vspace*{-25pt}
	\begin{equation}
	\label{eq:qual_symm}
	\Pr(\cT(x)=y)>0$ iff $\Pr(\cT(\pi(x))=y)>0.
	\end{equation}
	
	Now, if $x$ starts with either $\emptyset$ or $\{i_1,i_2\}$, then there is a single run on $x$ and on $\pi(x)$, namely $\sinit,s_\bot$, so both $x$ and $\pi(x)$ induce the same distribution on output sequences. Thus, \cref{eq:qual_symm} holds.
	
	Next, similarly to \cref{thm:deciding_approx_sym_undecidable}, we can again assume without loss of generality that $x$ starts with $\{i_1\}$, otherwise we use $\pi(x)$. Thus, $x$ is either of the form $\{i_1\}\cdot w$ or of the form $\{i_1\}\cdot w\cdot \{\{i_1\},\{i_2\}\}\cdot \Is$ with $w\in \{\emptyset,\{i_1,i_2\}\}^*$. 
	
	In the former case, recall that $\eta$ follows the transition function of $\cA$, as well as allowing at each point to reach $s_\bot$. Thus, $\cT(x)$ assigns positive probability to every word of the form $\emptyset^+\{o_1,o_2\}^*$ (of length $|w|+1$). Observe that $\pi(w)=w$, and hence $\pi(x)=\{i_2\}w$, which induces a distribution with the same support, and again \cref{eq:qual_symm} holds.
	
	In the latter case, $x$ is of the form $\{i_1\}\cdot w\cdot \{\{i_1\},\{i_2\}\}\cdot \Is$, where upon reading either $\{i_1\}$ or $\{i_2\}$, the runs in the $Q$ component all collapse to $s_\bot$. Thus, the support of $\cT(x)$ comprises words of the form $\emptyset^+\{o_1,o_2\}^*$ where the $\emptyset^+$ prefix is at most of length $|w|+1$. Since $\pi(\{i_1\})=\{i_2\}$ and $\pi(\{i_2\})=\{i_1\}$, then by the definition of $\eta$, the distribution $\cT(\pi(x))$ has the same support (as runs that remain in $\smid$ collapse to $s_\bot$ at the same stage). We thus conclude the claim.
	Finally, it is easy to see that the reduction is polynomial.
\end{proof}

\vspace*{-7pt}
Combining \cref{lem:qual_PSPACE,lem:qual_PSPACE_hard}, we have the following.
\begin{theorem}
	\label{thm:qual_PSPACE_complete}
	The problem of deciding, given an $I/O$ transducer $\cT$ and a permutation $\pi$, whether $\cT$ is $\pi$-qualitative-symmetric, is \PSPACE-complete.
\end{theorem}

As in \cref{sec:approx_sym}, since we use the permutation group $\cS_2$ for our hardness result, we have the following.
\begin{corollary}
\label{cor:qual_PSPACE_group}
	The problem of deciding whether a given $I/O$ transducer $\cT$ is $\pi$-qualitative-symmetric for every $\pi\in \cS_k$ is \PSPACE-complete.
\end{corollary}

\section{Extensions and Research Directions}
\paragraph*{Extensions}
The setting considered thus far restricts to corresponding input and output sets of the form $I=\{i_1,\ldots,i_k\}$ and $O=\{o_1,\ldots,o_k\}$. Typically, however, systems also include signals that are not process-specific, such as whether the system is ready, whether there is an error, etc. We can easily incorporate these into the setting. Indeed, adding input signals that are ignored by permutations can be inserted \emph{mutatis-mutandis} to all the automata constructions we use. In addition, the lower bounds trivially carry over. 

In addition, some systems have multiple sets of inputs and/or output signals that belong to processes, such as read grants and write grants, both of which are process-specific outputs. Again, our framework can easily be fit with this extension, by permuting each collection of process-specific inputs or outputs separately.

\vspace*{-10pt}
\paragraph*{Research Directions}
Process symmetry often arises in model checking, and exploiting it correctly can significantly reduce the size of specifications (and hence the time spent in model checking), as well as give insight into the behaviour of the system. In this work, we introduce several variants of process symmetry, and study their algorithmic aspects. Specifically, we show that exact symmetry can be decided in polynomial time, whereas the approximate version via the $L_\infty$ metric becomes undecidable. A coarser, qualitative approximation, can be decided in $\PSPACE$. In addition, a different type of symmetry, which looks only at the Parikh image of the output, can be decided efficiently.

The notions of symmetry studied in this work restrict to either letter-by-letter symmetry, or Parikh symmetry. However, many other directions can exploit the structure of words as temporal objects to define other symmetry measures. These include \emph{eventual symmetry}, where we require symmetry to take place only after a finite prefix, \emph{sliding-window symmetry}, where we look at Parikh images within a sliding window, while requiring window-by-window symmetry, as well as notions of symmetry that are only relevant for infinite words, such as the limit-average Parikh image.

\bibliography{main}

\begin{thebibliography}{10}

\bibitem{ball2008vacuity}
Thomas Ball and Orna Kupferman.
\newblock Vacuity in testing.
\newblock In {\em International Conference on Tests and Proofs}, pages 4--17.
  Springer, 2008.

\bibitem{cameron1999permutation}
Peter~J Cameron et~al.
\newblock {\em Permutation groups}, volume~45.
\newblock Cambridge University Press, 1999.

\bibitem{clarke1996exploiting}
Edmund~M. Clarke, Reinhard Enders, Thomas Filkorn, and Somesh Jha.
\newblock Exploiting symmetry in temporal logic model checking.
\newblock {\em Formal methods in system design}, 9(1-2):77--104, 1996.

\bibitem{clarke2018model}
Edmund~M Clarke~Jr, Orna Grumberg, Daniel Kroening, Doron Peled, and Helmut
  Veith.
\newblock {\em Model checking}.
\newblock MIT press, 2018.

\bibitem{donaldson2005symmetry}
A~Donaldson and Alice Miller.
\newblock Symmetry reduction for probabilistic systems.
\newblock In {\em Proc. 12th workshop on Automated Reasoning}, pages 17--18,
  2005.

\bibitem{emerson1996symmetry}
E~Allen Emerson and A~Prasad Sistla.
\newblock Symmetry and model checking.
\newblock {\em Formal methods in system design}, 9(1-2):105--131, 1996.

\bibitem{gimbert2010probabilistic}
Hugo Gimbert and Youssouf Oualhadj.
\newblock Probabilistic automata on finite words: Decidable and undecidable
  problems.
\newblock In {\em International Colloquium on Automata, Languages, and
  Programming}, pages 527--538. Springer, 2010.

\bibitem{ip1996better}
C~Norris Ip and David~L Dill.
\newblock Better verification through symmetry.
\newblock {\em Formal methods in system design}, 9(1-2):41--75, 1996.

\bibitem{kao2009nfas}
Jui-Yi Kao, Narad Rampersad, and Jeffrey Shallit.
\newblock On nfas where all states are final, initial, or both.
\newblock {\em Theoretical Computer Science}, 410(47-49):5010--5021, 2009.

\bibitem{kiefer2014stability}
Stefan Kiefer and Bj{\"o}rn Wachter.
\newblock Stability and complexity of minimising probabilistic automata.
\newblock In {\em International Colloquium on Automata, Languages, and
  Programming}, pages 268--279. Springer, 2014.

\bibitem{kwiatkowska2006symmetry}
Marta Kwiatkowska, Gethin Norman, and David Parker.
\newblock Symmetry reduction for probabilistic model checking.
\newblock In {\em International Conference on Computer Aided Verification},
  pages 234--248. Springer, 2006.

\bibitem{lin2016regular}
Anthony~W Lin, Truong~Khanh Nguyen, Philipp R{\"u}mmer, and Jun Sun.
\newblock Regular symmetry patterns.
\newblock In {\em International Conference on Verification, Model Checking, and
  Abstract Interpretation}, pages 455--475. Springer, 2016.

\bibitem{madani2003undecidability}
Omid Madani, Steve Hanks, and Anne Condon.
\newblock On the undecidability of probabilistic planning and related
  stochastic optimization problems.
\newblock {\em Artificial Intelligence}, 147(1-2):5--34, 2003.

\bibitem{paz2014introduction}
Azaria Paz.
\newblock {\em Introduction to probabilistic automata}.
\newblock Academic Press, 2014.

\bibitem{schutzenberger1961definition}
Marcel~Paul Sch{\"u}tzenberger.
\newblock On the definition of a family of automata.
\newblock {\em Inf. Control.}, 4(2-3):245--270, 1961.

\bibitem{sistla2000smc}
A~Prasad Sistla, Viktor Gyuris, and E~Allen Emerson.
\newblock Smc: a symmetry-based model checker for verification of safety and
  liveness properties.
\newblock {\em ACM Transactions on Software Engineering and Methodology
  (TOSEM)}, 9(2):133--166, 2000.

\bibitem{spermann2008prob}
Corinna Spermann and Michael Leuschel.
\newblock Prob gets nauty: Effective symmetry reduction for b and z models.
\newblock In {\em 2008 2nd IFIP/IEEE International Symposium on Theoretical
  Aspects of Software Engineering}, pages 15--22. IEEE, 2008.

\bibitem{tzeng1992polynomial}
Wen-Guey Tzeng.
\newblock A polynomial-time algorithm for the equivalence of probabilistic
  automata.
\newblock {\em SIAM Journal on Computing}, 21(2):216--227, 1992.

\bibitem{wahl2010replication}
Thomas Wahl and Alastair Donaldson.
\newblock Replication and abstraction: Symmetry in automated formal
  verification.
\newblock {\em Symmetry}, 2(2):799--847, 2010.

\end{thebibliography}

\section{Proofs}
\subsection{Proof of \cref{lem:symmetry_composition}}
	Consider $x\in \Ip$ and $y\in \Ip$, we wish to show that $\Pr(\cT(x)=y)=\Pr(\cT(\pi(\tau(x)))=\pi(\tau (y)))$. Since $\cT$ is $\tau$-symmetric, then $\Pr(\cT(x)=y)=\Pr(\cT(\tau(x))=\tau (y))$. 
	Next, since $\cT$ is $\pi$-symmetric, then applying the definition for the input $\tau(x)\in \Ip$ and $\tau(y)\in \Op$, we have that $\Pr(\cT(\tau(x))=\tau (y))= \Pr(\cT(\pi(\tau(x)))=\pi(\tau(y)))$, and so overall $\Pr(\cT(x)=y)=\Pr(\cT(\pi(\tau(x)))=\pi(\tau(y)))$ and we are done.
\qed

\subsection{Proof of \cref{lem:tran_to_PRA}}
	The translation is similar to the one given in the proof of \cref{thm:deciding_symmetry_permutation_ptime}, where instead of adding $\tO$ to the alphabet, we collate the Parikh image using the rewards. 
	
	Let $\cT=\tup{I,O,S,s_0,\delta,\lab}$, we construct $\cA=\tup{S,\tI,\delta,s_0,S}$ with the following reward function: for every $s\in S$ and $1\le j\le k$, we have $\rew(s)_j=1$ if $o_j\in \lab(s)$ and $\rew(s)_j=0$ otherwise (that is, $\rew(s)$ is the characteristic vector of $\lab(s)$). Thus, $\cA$ is identical to $\cT$, where we treat all states as accepting, and replace output labels with their characteristic vectors.
	
	The construction of $\cB$ is similar, but accounts for the permutation $\pi$: we define $\cB=\tup{S,\tI,\mu,s_0,S}$ with reward function $\rew'$, where $\mu(s,\vec{i})=\delta(s,\pi(\vec{i}))$ for every state $s\in S$ and $\vec{i}\in \tI$, and $\rew'(s)=\pi(\rew(s))$ (where $\rew$ is the reward function of $\cA$). It is easy to see that the construction of $\cA$ and $\cB$ satisfies the conditions of the lemma.

\subsection{Proof of \cref{lem:qual_PSPACE}}
	Similarly to our approach in \cref{thm:deciding_symmetry_permutation_ptime}, we translate $\cT$ to two automata $\cA$ and $\cB$, where $\cA$ mimics the operation of $\cT$, and $\cB$ works similarly, but under the permutation $\pi$. Then, we check the equivalence of $\cA$ and $\cB$. Instead of using PAs, however, we now use nondeterministic automata (NFAs). 
	An NFA is $\cN=\tup{Q,\Sigma,\delta,q_0,F}$ where $Q$ is a set of states, $\Sigma$ is an alphabet, $\delta:Q\times \Sigma\to 2^Q$ is a transition function, $q_0$ is an initial state, and $F$ are the accepting states. The semantics of NFAs are textbook standard.
	
	Let $\cT=\tup{I,O,S,s_0,\delta,\lab}$. We define $\cA=\tup{S,\tIO,\eta,s_0,S}$ and $\cB=\tup{S,\tIO,\zeta,s_0,S}$, where the transition functions are defined as follows. 
	Let $q\in S$ and $\sigma=\vec{i}\cup\vec{o}$ with $\vec{i}\in \tI$ and $\vec{o}\in \tO$, then 
	$\eta(q,\sigma)=\{p\in S\ST \delta(q,\vec{i})(p)>0 \mbox{ and }\lab(p)=\vec{o}\}$ and 	$\zeta(q,\sigma)=\{p\in S\ST \delta(q,\pi(\vec{i}))(p)>0 \mbox{ and }\lab(p)=\pi(\vec{o})\}$.
	
	By construction, for every $x\in \Ip$ and $y\in \Op$ we have that $\Pr(\cT(x)=y)>0$ iff $\cA$ accepts $x\otimes y$, and $\Pr(\cT(\pi(x))=\pi(y))$ iff $\cB$ accepts $x\otimes y$. Thus, we have that $\cT$ is $\pi$-qualitative-symmetric iff $L(\cA)=L(\cB)$. Since equivalence of NFAs can be checked in \PSPACE, we are done.
	
\subsection{Correctness proof of \cref{thm:deciding_approx_sym_undecidable}}
\label{apx:correctness}
For the first direction, let $x'\in \{\emptyset,\{i_2\}\}^*$ such that $\cA(x')> \lambda$, and consider the word $x=\{i_2\}\cdot x'\cdot \{i_1,i_2\}$. By the construction of $\cT$, after seeing $\{i_2\}$, there is only a single run of $\cT$ which proceeds to $q_0$. From there, $\cT$ mimics the behaviour of $\cA$ on $x'$. Thus, after reading $x'$, the distribution of states has probability $\cA(x)$ for states in $F$, and probability $1-\cA(x)$ in states in $Q\setminus F$. Note that up until then, only the label $\emptyset$ is seen, so the distribution of outputs is $\dirac{\emptyset^{|x'|+1}}$. Then, after reading $\{i_1,i_2\}$, the distribution of outputs give probability $\cA(x)$ to $\emptyset^{|x'|+2}$, and $1-\cA(x)$ to $\emptyset^{|x'|+1}\cdot \{o_1,o_2\}$. 

Now consider $\pi(x)=\{i_1\}\cdot \pi(x')\cdot \{i_1,i_2\}$. Upon reading $\{i_1\}$, the single run of $\cT$ arrives at $\smid$. Then, since $x'\in \{\emptyset,\{i_2\}\}^*$, we have that $\pi(x')\in \{\emptyset,\{i_1\}\}^*$, so the run of $\cT$ stays in $\smid$. Finally, reading $\{i_1,i_2\}$, the run moves to $s_{\bot}$. Therefore $\cT(x)$ gives probability 1 to the output $\emptyset^{|x'|+1}\{o_1,o_2\}$. Thus, for the output $y=\emptyset^{|x'|+2}$, we have that $|\Pr(\cT(x)=y)-\Pr(\cT(\pi(x))=y)|=|\cA(x)-0|>\lambda=\epsilon$, so $\cT$ is not $(\epsilon,\pi)$-symmetric.

For the converse direction, assume $x,y$ are such that $|\Pr(\cT(x)=y)-\Pr(\cT(\pi(x))=y)|>\epsilon$. 
We start by eliminating candidates for such $x$ and $y$. 
First, observe that if $x$ starts with $\emptyset$ or $\{\i_1,\o_1\}$ (both of which are invariant under $\pi$), we have $\cT(x)$ gives probability $1$ to the output $\lab(q_{\bot})^{|x|}=\{o_1,o_2\}^{|x|}$, and so $\cT(x)=\cT(\pi(x))$, hence $|\Pr(\cT(x)=y)-\Pr(\cT(\pi(x))=y)|=0$ for all $y$, so this case cannot occur.

Next, we claim that without loss of generality, we can assume $x$ starts with $\{i_2\}$. Indeed, if $x$ starts with $\{i_1\}$, then $\pi(x)$ starts with $\{i_2\}$. Since $\pi(\pi(x))=x$, we could start the argument with $\pi(x)$, while maintaining \cref{eq:undec_eq}.

Now, if $x$ is of the form $\{i_2\}\cdot \{\emptyset,\{i_2\}\}^n$, then $\cT(x)$ gives probability 1 to the output $\emptyset^{n+1}$, but $\pi(x)$ is now of the form $\{i_1\}\cdot \{\emptyset,\{i_1\}\}^n$, which also induces the same distribution, this case cannot occur as well.

It follows that $x$ is of the form $\{i_2\}\cdot x' \cdot \{\{i_1\},\{i_1,i_2\}\}\cdot \Is $ where $x'\in \{\emptyset,\{i_2\}\}^n$. 
We claim that $\cA(x')>\lambda$. Indeed, as we observed above, $\cT(x)$ gives probability $\cA(x')$ to the output $\emptyset^{|x|}$ and probability $1-\cA(x')$ to the output $\emptyset^{|x'|+1}\cdot \{o_1,o_2\}^{|x|-|x'|-1}$. However, $\cT(\pi(x))$ gives probability $1$ to the output $\emptyset^{|x'|+1}\cdot \{o_1,o_2\}^{|x|-|x'|-1}$. Thus, there are only two possibilities for $y$ in order for \cref{eq:undec_eq} to hold: if $y=\emptyset^{|x|}$, we have
\[
\lambda=\epsilon< |\Pr(\cT(x)=y)-\Pr(\cT(\pi(x))=y)|=|\cA(x')-0|=\cA(x')
\]
and if $y=\emptyset^{|x'|+1}\cdot \{o_1,o_2\}^{|x|-|x'|-1}$, then 
\[
\lambda=\epsilon< |\Pr(\cT(x)=y)-\Pr(\cT(\pi(x))=y)|=|1-\cA(x')-1|=\cA(x')
\]	
So in either case $\cA(x')>\lambda$, and we are done.

\end{document}